\pdfoutput=1
\documentclass[12pt]{article}
\usepackage[utf8]{inputenc}
\usepackage[british]{babel}
\usepackage{cmap}
\usepackage{lmodern}

\usepackage{amssymb, amsmath, amsthm}
\usepackage[a4paper,top=25mm,bottom=25mm,left=25mm,right=25mm]{geometry}
\usepackage{ragged2e}

\usepackage{authblk} 
\usepackage{pifont}
\usepackage{graphicx}
\usepackage[dvipsnames,svgnames,table]{xcolor}
\usepackage[figuresright]{rotating}
\usepackage{xtab} 
\usepackage{longtable} 
\usepackage{multirow}
\usepackage{footnote}
\usepackage[stable]{footmisc}
\usepackage{chngpage} 
\usepackage{pdflscape} 
\usepackage[nottoc,notlot,notlof]{tocbibind} 

\usepackage{pgfplots}
\pgfplotsset{every tick label/.append style={font=\footnotesize}}
\pgfplotsset{compat=1.18}
\usepackage{setspace}

\usepackage{array}
\newcolumntype{K}[1]{>{\centering\arraybackslash$}p{#1}<{$}}

\makesavenoteenv{tabular}
\usepackage{tabularx}
\usepackage{booktabs}
\usepackage{threeparttable}
\usepackage[referable]{threeparttablex} 
\newcolumntype{R}{>{\raggedleft\arraybackslash}X}
\newcolumntype{L}{>{\raggedright\arraybackslash}X}
\newcolumntype{C}{>{\centering\arraybackslash}X}
\newcolumntype{A}{>{\columncolor{gray!25}}C}
\newcolumntype{a}{>{\columncolor{gray!25}}c}

\newlength{\tablen}

\usepackage{dcolumn} 
\newcolumntype{.}{D{.}{.}{-1}}

\usepackage{tikz}
\usetikzlibrary{arrows, calc, matrix, patterns, positioning, trees}
\usepackage[semicolon]{natbib} 
\usepackage[hyphens]{url}
\usepackage{microtype}
\usepackage[justification=centering]{caption} 

\usepackage[labelformat=simple]{subcaption}

\DeclareCaptionLabelFormat{parenthesis}{(#2)}
\makeatletter
\renewcommand\p@subfigure{\arabic{figure}.}
\makeatother

\DeclareCaptionLabelFormat{parenthesis}{(#2)}
\makeatletter
\renewcommand\p@subtable{\arabic{table}.}
\makeatother

\usepackage{enumitem}

\setlist[itemize]{leftmargin=2.5\parindent}
\setlist[enumerate]{leftmargin=2.5\parindent}

\usepackage{hyperref} 
\hypersetup{
  colorlinks   = true,    		
  urlcolor     = blue,    		
  linkcolor    = blue,    		
  citecolor    = ForestGreen	
}

\def\addlegendimage{\csname pgfplots@addlegendimage\endcsname}

\theoremstyle{plain}

\newtheorem{proposition}{Proposition}

\theoremstyle{definition}

\newtheorem{example}{Example}

\theoremstyle{remark}

\makeatletter
\let\@fnsymbol\@alph
\makeatother

\def\keywords{\vspace{.5em} 
{\noindent \textit{Keywords}: }}

\def\AMS{\vspace{.5em} 
{\noindent \textbf{\emph{MSC} class}: }}

\def\JEL{\vspace{.5em} 
{\noindent \textbf{\emph{JEL} classification number}: }}

\title{The allocation of FIFA World Cup slots \\ based on the ranking of confederations}

\author{\href{https://sites.google.com/view/laszlocsato}{L\'aszl\'o Csat\'o}\thanks{~Corresponding author. Email: \emph{laszlo.csato@sztaki.hun-ren.hu} \newline
Institute for Computer Science and Control (SZTAKI), Hungarian Research Network (HUN-REN), Laboratory on Engineering and Management Intelligence, Research Group of Operations Research and Decision Systems, Budapest, Hungary \newline
Corvinus University of Budapest (BCE), Institute of Operations and Decision Sciences, Department of Operations Research and Actuarial Sciences, Budapest, Hungary}
$\qquad \qquad$
L\'aszl\'o Marcell Kiss\thanks{~Email: \emph{laszlo.kiss2@stud.uni-corvinus.hu} \newline
Corvinus University of Budapest (BCE), Budapest, Hungary}
$\qquad \qquad$
\href{https://sites.google.com/view/zsomborszadoczki/}{Zsombor Sz\'adoczki}\thanks{~Email: \emph{zsombor.szadoczki@sztaki.hun-ren.hu} \newline
Institute for Computer Science and Control (SZTAKI), Hungarian Research Network (HUN-REN), Laboratory on Engineering and Management Intelligence, Budapest, Hungary \newline
Corvinus University of Budapest (BCE), Institute of Operations and Decision Sciences, Department of Operations Research and Actuarial Sciences, Budapest, Hungary}
}

\date{\today}

\def\Dedication{
{\noindent
``\emph{Render unto Caesar the things that are Caesar's, and unto God the things that are God's}''
}
\flushright
\noindent (Matthew 22:21)
\vspace{0.5cm} 
\justify }

\begin{document}

\maketitle
\thispagestyle{empty}
\Dedication

\begin{abstract}
\noindent
Qualifications for several world championships in sports are organised such that distinct sets of teams play in their own tournament for a predetermined number of slots. Inspired by a recent work studying the problem with the tools from the literature on fair allocation, this paper provides an alternative approach based on historical matches between these sets of teams. We focus on the FIFA World Cup due to the existence of an official rating system and its recent expansion to 48 teams, as well as to allow for a comparison with the already suggested allocations. Our proposal extends the methodology of the FIFA World Ranking to compare the strengths of five confederations. Various allocations are presented depending on the length of the sample, the set of teams considered, as well as the frequency of rating updates. The results show that more European and South American teams should play in the FIFA World Cup. The ranking of continents by the number of deserved slots is different from the ranking implied by FIFA policy. We recommend allocating at least some slots transparently, based on historical performances, similar to the access list of the UEFA Champions League.


\AMS{90-10, 90B90, 91B14}

\JEL{C44, D71, Z20}
\end{abstract}

\clearpage

\section{Introduction} \label{Sec1}

In most sports tournaments, the number of competitors is limited. Thus, if there is no objective measure of performance such as finishing time, the allocation of these qualifying slots is far from trivial. The current paper aims to suggest an approach to that end on the basis of pairwise comparisons between different sets of teams.

In particular, we will pick up the FIFA World Cup as a case study.
Every four years, FIFA (F\'ed\'eration internationale de football association), the governing body of (association) football---the most played sport around the world---organises the FIFA World Cup. This is one of the most popular sporting events in the world; around 5 billion people have engaged in the 2022 FIFA World Cup, which has generated 93.6 million posts across all social media platforms with 262 billion cumulative reach \citep{FIFA2023a}. The FIFA World Cup also attracts many visitors from around the world who are keen to visit the host nation and the neighboring states, which generates a powerful demand for providers of tourist services \citep{ManciniTrikiPiya2022}.

The FIFA World Cup qualification is mainly played within the six FIFA continental zones: AFC (Asia), CAF (Africa), CONCACAF (North and Central America and the Caribbean), CONMEBOL (South America), OFC (Oceania), and UEFA (Europe). For each World Cup, FIFA decides the number of places given to the six confederations. However, the allocation rule is fully obscure \citep{StoneRod2016} and unfair \citep{Csato2023c}, even though the issue has recently received much attention due to the expansion to 48 teams from 2026 \citep{KrumerMoreno-Ternero2023}.

The suggested allocation rule essentially relies on a ranking of these six continents. In particular, we modify the calculation formula of the official FIFA World Ranking, used to measure the strengths of national teams, to quantify the performance of each confederation in the previous FIFA World Cups and their qualifications. Although one might think such an extension is straightforward, it requires considering a lot of factors as detailed in Section~\ref{Sec4}. Finally, Proposition~\ref{Prop1} uncovers that the ratios of the win expectancies derived from Elo ratings are transitive, thus, they can be directly used to proportionally allocate the number of slots available.

It is found that more European and South American teams should play in the 2026 FIFA World Cup than allowed by FIFA. In addition, the CONCACAF quota needs to be higher than the AFC quota, which is in stark contrast with the policy of the previous decades. The records of African teams turn out to be the most volatile and unpredictable.

The methodology proposed here can be used to allocate the qualifying slots in a transparent manner. Compared to \citet{KrumerMoreno-Ternero2023}, an important novelty of our study is the ``look into the past'': historical slot allocations are also presented assuming that the expansion to 48 teams would have been implemented before. Furthermore, since the suggestion is based on a reasonable extension of the official FIFA World Ranking, it might be more easily accepted by the stakeholders (officials, players, coaches, TV broadcasters, and fans) than the complex methods of fair allocation proposed by \citet{KrumerMoreno-Ternero2023}, which depend to a great extent on how proportionality is defined.

It must be noted that FIFA wants to achieve several aims with the World Cup other than guaranteeing the fairness of slot allocation, for example, opening new markets and maximising ticket sales. Shortly, the FIFA World Cup can serve at least four different economic goals \citep{Kurscheidt2006}:
(1) it is the financing basis for the global football business;
(2) it is used to attract corporate sponsors, in particular, companies from the industry of sports goods;
(3) it offers a unique opportunity for professionals, especially those who are not playing in major leagues, to accelerate their careers; and
(4) it is promoted as contributing to the development of the host nations and cities.

Furthermore, since broadcasting revenue depends on the participating teams, increasing the chances of the countries with the largest markets (e.g.\ China) may substantially contribute to the value of broadcasting rights. The allocation of higher revenue---an issue that has recently received serious attention in the literature \citep{BergantinosMoreno-Ternero2020a, BergantinosMoreno-Ternero2023d, BergantinosMoreno-Ternero2023c}---could eventually help achieve other crucial targets such as promoting football in less developed nations, and improving fairness even if the slot allocation is not exclusively determined by historical performances.

Therefore, some other criteria of justice could explain the allocation policy of FIFA, and we do not intend to state that the proposed meritocratic allocation rule should be chosen. However, the number of slots ``deserved'' by each continental zone is an important element of fairness as the beginning quote shows. Our results certainly provide useful information for the decision-makers when debating about the allocation of FIFA World Cup slots.

The paper is structured as follows. Section~\ref{Sec2} provides a concise literature review. The underlying data are described in Section~\ref{Sec3}, while the methodology is detailed in Section~\ref{Sec4}. Section~\ref{Sec5} presents the results, and Section~\ref{Sec6} offers some concluding remarks.

\section{Related literature} \label{Sec2}

Although the allocation of scarce resources is a fundamental problem of economics \citep{Thomson2019}, there are remarkably few papers on allocating the FIFA World Cup berths.
According to \citet{StoneRod2016}, the current system of qualification is based neither on ensuring the participation of the best 32 teams in the world, nor does it fairly allocate slots by the number of teams per confederation or any other reasonable metric. Indeed, \citet{Csato2023c} finds serious differences between the continents: for instance, a South American team could have tripled its chances of qualification by playing in Asia. Analogously, the move of Australia from the Oceanian to the Asian zone has increased its probability of participating in the 2018 FIFA World Cup by about 65\%.

Inspired by the recent expansion to 48 teams, \citet{KrumerMoreno-Ternero2023} explore the allocation of additional slots among continental confederations by using the standard tools of the fair allocation literature. The ``claims'' of the continents are based on the FIFA World Ranking and the World Football Elo Ratings (\url{http://eloratings.net/}) that are summed for all member countries or just for countries being in the top 48. They also consider the average annual number of teams in the top 31 and the average annual number of teams ranked 32--48. In contrast, our approach exclusively depends on the results of the national teams that have played in the FIFA World Cup and its inter-continental play-offs. Therefore, neither friendlies, nor matches played in continental championships and qualifications affect the allocation of FIFA World Cup slots proposed here since these games provide either unreliable (friendlies) or no information on the relative strengths of the regions.

On the other hand, the tournament format of the 2026 FIFA World Cup has been extensively investigated.
Both \citet{Guyon2020a} and \citet{ChaterArrondelGayantLaslier2021} have strongly criticised the original choice of FIFA in 2017 (16 groups of three teams each), although \citet{Truta2018} argues that this format would not have increased the number of non-competitive matches. FIFA has finally decided for 12 groups of four teams each, an option recommended by \citet{Guyon2020a} and studied by \citet{ChaterArrondelGayantLaslier2021}. There are further proposals: \citet{Renno-Costa2023} presents a double-elimination structure instead of the group stage in order to produce more competitive and exciting matches, while \citet{GuajardoKrumer2023} develop three alternative formats and schedules for a FIFA World Cup with 12 groups of four teams to meet several important criteria of fairness.

Our allocation rule follows the methodology of the official FIFA World Ranking, which uses an Elo-based approach since 2018 \citep{FIFA2018c}. \citet{GomesdePinhoZancoSzczecinskiKuhnSeara2024} present a comprehensive stochastic analysis of the Elo algorithm.
The 2018 reform has corrected several weaknesses of the previous FIFA World Ranking that are summarised by \citet{CeaDuranGuajardoSureSiebertZamorano2020, Csato2021a, Kaminski2022, LasekSzlavikGagolewskiBhulai2016}. However, \citet{SzczecinskiRoatis2022} still identify some possible improvements with respect to predictive capacity:
(a) the importance of games defined by FIFA is counterproductive;
(b) home field advantage should be included;
(c) the results can be weighted with the goal differential;
(d) the shootout and knockout rules need to be removed as they are not rooted in any solid statistical principle.
Nonetheless, even though the World Football Elo Ratings take home advantage and goal difference into account and are a good indicator of football success at the international level \citep{LasekSzlavikBhulai2013, GasquezRoyuela2016}, we will retain the current formula of the FIFA World Ranking since this might be more acceptable for the decision-makers.




\section{Data} \label{Sec3}

We have collected all matches played in the FIFA World Cups and their inter-continental play-offs since the 1954 edition.
Two inter-continental play-offs are disregarded:
\begin{itemize}
\item
In the \href{https://en.wikipedia.org/wiki/1958_FIFA_World_Cup_qualification_(AFC/CAF\%E2\%80\%93UEFA_play-off)}{1958 FIFA World Cup qualification (AFC/CAF--UEFA play-off)}, played by Israel and Wales, Israel had advanced from the AFC/CAF qualification zone without playing any match due to the withdrawal of several teams for political reasons.
\item
In the \href{https://en.wikipedia.org/wiki/1974_FIFA_World_Cup_qualification_(UEFA\%E2\%80\%93CONMEBOL_play-off)}{1974 FIFA World Cup qualification (UEFA--CONMEBOL play-off)}, played by the Soviet Union and Chile, the Soviet Union refused to play in the second leg that was held in a stadium where some people were tortured and killed just two weeks before the match.
\end{itemize}
1954 was chosen as the starting date because some teams withdrew from the 1950 FIFA World Cup, and the previous two competitions (1942, 1946) were cancelled due to the second world war. A recent paper has also analysed competitive imbalance between FIFA World Cup groups from 1954 through 2022 \citep{LaprePalazzolo2023}.
Each country has been classified to the confederation where it had played at the time of the match. This is especially important for Australia, which played in the OFC zone until the 2006 FIFA World Cup but in the AFC zone since then. Another interesting fact is that Israel won the 1990 FIFA World Cup qualification for the Oceanian zone and, thus, qualified for the inter-continental play-off against Colombia.

\subsection{Ignoring the Oceanian (OFC) zone}

National teams from Oceania played in the Asian (AFC) zone until the 1982 FIFA World Cup. For the FIFA World Cups from 1986 to 2022, the OFC had a ``half'' slot, that is, the winner of its qualification tournament participated in the inter-continental play-offs. This country was Australia between 1986 and 2006, except for 1990 (when it was Israel), and New Zealand since 2010. Australia has moved to the AFC zone from the 2010 FIFA World Cup. Consequently, at most one Oceanian team played in any FIFA World Cup, and it was usually a country that is currently outside the OFC. Therefore, in contrast to \citet{KrumerMoreno-Ternero2023}, we have decided to fix the OFC quota at $4/3$ as determined by FIFA for the 2026 World Cup. However, most rules suggested by \citet{KrumerMoreno-Ternero2023} imply a lower number of slots deserved.

\subsection{Consideration of inter-continental play-offs}

Inter-continental play-offs are usually organised as two-legged home-and-away ties. Since the qualification is determined by the aggregated result over the two legs, a team may be satisfied with a draw in the second leg. For example, a 1-1 draw was sufficient for Uruguay against Costa Rica to play in the 2010 FIFA World Cup. This potential problem of incentives can be treated by considering the play-off as one match, possibly with a higher weight because playing more matches favours the stronger team in general \citep{LasekGagolewski2018}. This is done by the Football Club Elo Ratings (\url{http://clubelo.com/}), a website providing Elo ratings for European club football, where the weight of these two-legged clashes is $\sqrt{2}$ \citep{Csato2022b}.

The FIFA World Ranking treats two-legged play-offs as two separate matches.

\subsection{Disregarding the last round of group matches}

In the last round of the group stage, some already qualified teams may play with little enthusiasm and take into account other factors such as resting their best players \citep{ChaterArrondelGayantLaslier2021}. These matches can offer opportunities for collusion \citep{KendallLenten2017, Guyon2020a}, or even for tanking (deliberately losing to face a preferred opponent in the next stage).
As an illustration, take Group D in the 2022 FIFA World Cup. Here, France beat Australia and Denmark in the first two rounds, which ensured its qualification for the knockout stage. France won the group despite losing against Tunisia in the third round.

The second group stage, organised in the 1974, 1978, and 1982 FIFA World Cups, suffers from this problem of incentives to a lesser extent because (1) the two group winners advanced to the final and the two runners-up advanced to the third-place game in 1974 and 1978; and (2) the four group winners advanced to the semifinals in 1982.

Therefore, we have decided to disregard the last round of group matches---except for the second group stage---in the baseline in order to ignore the potential problem of incentives. Even though that solution is not perfect (for instance, a draw was equivalent to a win for the Netherlands against Brazil in the 1974 FIFA World Cup Group A), identifying all these matches would be a cumbersome procedure without a substantial benefit.

\subsection{Descriptive statistics}

The number of national teams playing in the FIFA World Cup was 16 between 1954 and 1978 (7 editions), 24 from 1982 to 1994 (4 editions), and 32 since 1998 (7 editions). The tournament always started with a group stage played in (4/6/8) groups of four teams each. In 1974, 1978, and 1982, there was a second group stage with two groups of four (1974, 1978), or four groups of three (1982). Usually, 8 (until 1970) or 16 (since 1986) teams qualified for the knockout phase, except for the three events with a second group stage, where the knockout phase consisted of only the final and the third-place game (1974, 1978) or two semifinals, the final, and the third-place game (1982).
\citet[Table~1]{LaprePalazzolo2023} overviews the formats of all FIFA World Cups.

\begin{table}[t!]
\centering
  \caption{The set of matches by pairs of confederations}
  \label{Table1}
\centerline{
\begin{threeparttable}
    \rowcolors{1}{}{gray!20}
\begin{small}
\begin{tabularx}{1.15\textwidth}{lCCCCCCCCCCCCCCCCCCc} \toprule
&	54&	58&	62&	66&	70&	74&	78&	82&	86&	90&	94&	98&	02&	06&	10&	14&	18&	22 & $\sum$  \\ \bottomrule
    AFC--CAF & 0     & 0     & 0     & 0     & 0     & 0     & 0     & 0     & 0     & 0     & 1     & 0     & 1     & 2     & 2     & 3     & 2     & 3     & 14 \\
    AFC--CONC & 0     & 0     & 0     & 0     & 0     & 0     & 0     & 0     & 0     & 0     & 0     & 2     & 2     & 1     & 0     & 0     & 1     & 1     & 7 \\
    AFC--CONM & 0     & 0     & 0     & 1     & 1     & 0     & 0     & 0     & 2     & 1     & 1     & 1     & 1     & 0     & 4     & 2     & 2     & 5     & 21 \\
    AFC--UEFA & 2     & 0     & 0     & 2     & 1     & 0     & 2     & 2     & 2     & 3     & 3     & 5     & 9     & 4     & 4     & 3     & 6     & 6     & 54 \\
    CAF--CONC & 0     & 0     & 0     & 0     & 0     & 0     & 1     & 0     & 0     & 0     & 0     & 0     & 0     & 1     & 2     & 2     & 0     & 0     & 6 \\
    CAF--CONM & 0     & 0     & 0     & 0     & 1     & 0     & 0     & 1     & 1     & 2     & 2     & 2     & 2     & 2     & 4     & 1     & 1     & 0     & 19 \\
    CAF--UEFA & 0     & 0     & 0     & 0     & 1     & 2     & 1     & 3     & 4     & 4     & 4     & 9     & 9     & 6     & 6     & 6     & 7     & 12    & 74 \\
    CONC--CONM & 1     & 0     & 1     & 0     & 0     & 0     & 0     & 0     & 1     & 1     & 2     & 1     & 1     & 2     & 2     & 3     & 2     & 1     & 18 \\
    CONC--UEFA & 1     & 2     & 1     & 2     & 3     & 2     & 1     & 4     & 5     & 4     & 4     & 4     & 4     & 5     & 4     & 7     & 4     & 7     & 64 \\
    CONM--UEFA & 7     & 9     & 11    & 9     & 7     & 13    & 11    & 9     & 10    & 9     & 8     & 12    & 11    & 8     & 9     & 12    & 11    & 8     & 174 \\ \hline
    Playoffs (1 leg) & 0     & 0     & 0     & 0     & 0     & 0     & 0     & 0     & 0     & 0     & 0     & 0     & 0     & 0     & 0     & 0     & 0     & 1     & 1 \\
    Playoffs (2 legs) & 0     & 0     & 6     & 0     & 0     & 0     & 1     & 0     & 0     & 0     & 0     & 0     & 1     & 1     & 1     & 1     & 1     & 0     & 12 \\ \hline
    $\sum$ & 11    & 11    & 19    & 14    & 14    & 17    & 17    & 19    & 25    & 24    & 25    & 36    & 41    & 32    & 38    & 40    & 37    & 44    & \textbf{464} \\ \toprule
\end{tabularx}
\end{small}
\begin{tablenotes} \footnotesize
\item
Years of FIFA World Cups are abbreviated by the last two digits.
\item
CONC stands for CONCACAF; CONM stands for CONMEBOL.
\item
Matches played in the last round of the (first) group stage are not included.
\end{tablenotes}
\end{threeparttable}
}
\end{table}

Table~\ref{Table1} presents the number of matches between different confederations in our database that has been collected from \url{https://www.kaggle.com/datasets/piterfm/fifa-football-world-cup?select=matches_1930_2022.csv}. The inter-continental play-offs have been added manually. Unsurprisingly, most of the games involve a European team: 81\% of the 451 inter-continental matches played in the FIFA World Cups since 1954. More than one-third of the matches have been played by a CONMEBOL team against a UEFA team as these continents are the most successful. On the other hand, there are less than 10 matches between AFC and CONCACAF (two further two-legged clashes can be found between these continents in the play-offs), as well as between CAF and CONCACAF.
 
\subsection{Historical slot allocations}

\begin{table}[t!]
  \centering
  \caption{Distribution of slots in the FIFA World Cups with at least 32 teams}
  \label{Table2}
\begin{threeparttable}
    \rowcolors{1}{}{gray!20}
    \begin{tabularx}{\textwidth}{l CCC CCC Cc} \toprule
    Confederation & 1998  & 2002  & 2006  & 2010  & 2014  & 2018  & 2022  & 2026 \\ \bottomrule
    AFC   & 3.5   & 2.5+2 & 4.5   & 4.5   & 4.5   & 4.5   & 4.5+1   & 8.33 \\
    CAF   & 5     & 5     & 5     & 5+1   & 5     & 5     & 5     & 9.33 \\
    CONCACAF & 3     & 3     & 3.5   & 3.5   & 3.5   & 3.5   & 3.5   & 3.67+3 \\
    CONMEBOL & 5     & 4.5   & 4.5   & 4.5   & 4.5+1   & 4.5   & 4.5   & 6.33 \\
    OFC   & 0.5   & 0.5   & 0.5   & 0.5   & 0.5   & 0.5   & 0.5   & 1.33 \\
    UEFA  & 14+1  & 14.5  & 13+1  & 13    & 13    & 13+1  & 13    & 16 \\ \bottomrule
    Total & 32    & 32    & 32    & 32    & 32    & 32    & 32    & 48 \\ \toprule
    \end{tabularx}
\begin{tablenotes} \footnotesize
\item
Fractions represent slots available through the play-offs.
\item
The number after the $+$ sign indicates host(s).
\end{tablenotes}
\end{threeparttable}
\end{table}

Table~\ref{Table2} shows the places allocated for the six continents from 1998. The distribution was fixed from 2006 until 2022, an additional spot varied due to the region of the host country. The expansion in 2026 has favoured mainly the Asian and African zones in absolute terms, while Europe and South America have benefited the least in relative terms.

\section{Methodology} \label{Sec4}

We aim to follow the calculation formula of the official FIFA World Ranking \citep{FIFA2018c} to the extent possible. This has been constructed to quantify the strength of national teams but can be modified to measure the strengths of sets of teams as well.

\subsection{Preliminaries} \label{Sec41}

Let us overview the FIFA World Ranking \citep{FIFA2018c}, which is essentially an Elo-based method.

The rating $R_i^0$ of team $i$ is updated to $R_i^1$ after it plays a match against team $j$ as follows:
\begin{equation} \label{eq_Elo_update}
R_i^{(1)} = R_i^{(0)} + \Delta R_i = R_i^{(0)} + I \times \left( W - W_{ij}^E \right),  
\end{equation}
where $I$ is the importance, $W$ is the result, and $W_{ij}^E$ is the expected result of the match.

Regarding match importance, there are three types of games in our data \citep{FIFA2018c}:
(a) $I=25$ for FIFA World Cup qualification matches (inter-continental play-offs);
(b) $I=50$ for FIFA World Cup matches until the Round of 16;
(c) $I=60$ for FIFA World Cup matches from the quarterfinals onwards.

There was a second group stage instead of quarterfinals between 1974 and 1982 with eight teams in 1974, 1978, and 12 teams in 1982. The importance of matches played in this second group stage is chosen to be 60 in 1974 and 1978, but 50 in 1982.

The result of any match can be a win ($W=1$), a draw ($W=0.5$), or a loss ($W=0$).
In the knockout stage, each match should have a winner. Matches decided after extra time are treated accordingly. In the case of penalty shootouts, the match becomes a draw for the losing team ($W=0.5$) and ``half a win'' for the winning team ($W=0.75$).

$W_e$ is determined by the following equation:
\begin{equation} \label{eq_win_probability}
W_{ij}^E = \frac{1}{1+10^{- \left( R_i - R_j \right) / 600}}.
\end{equation}
Note that $W_{ji}^E = 1- W_{ij}^E$. $W_{ij}^E > 0.5$ if and only if $R_i > R_j$, hence, a draw is advantageous for the weaker and unfavourable for the stronger team. $600$ is a scaling factor.

There is also a special rule: teams cannot earn negative points in the knockout round of the FIFA World Cup as a result of losing, or winning a penalty shootout against a weaker team.

Formula~\eqref{eq_Elo_update} ensures that the sum of Elo ratings over all teams does not change since $\Delta R_i + \Delta R_j = 0$ for any match played by teams $i$ and $j$.
However, the treatment of matches decided by a penalty shootout and the prohibition of losing points in the knockout stage imply a constant ``inflation'' in the ratings. Theoretically, this might imply a ``vicious'' circle: confederations with many delegates have a higher probability to play in the knockout stage, where their Elo cannot decrease, thus, their dominance becomes guaranteed. Consequently, it would be useful to remove both the knockout and the shootout rules that
are not rooted in any solid statistical principle \citep{SzczecinskiRoatis2022}.

An online calculator for the FIFA rating is available at \url{https://hermann-baum.de/excel/WorldCup/en/FIFA_Ranking.php}, where the corresponding Excel solution can be downloaded, too.

\subsection{The frequency of updating}

The above method of the FIFA World Ranking can be directly used to evaluate the performance of any set of teams, albeit this implies a novel challenge: choosing the frequency of updating. A given team very rarely plays two matches in less than three days, and cannot play two matches simultaneously by definition. On the other hand, if one considers, for example, Argentina and Brazil (two CONMEBOL nations) might play against different opponents on the same day or even at the same time. Section~\ref{Sec43} will present later that the sequence of matches has a non-negligible effect on the Elo ratings---thus, it \emph{does} count in which order the games of Argentina and Brazil are taken into account.

Therefore, three options will be considered:
\begin{enumerate}
\item
\emph{Updating between rounds}: The Elo ratings of the continents are fixed at the beginning of each round of group matches, as well as at the beginning of each knockout round. The inter-continental play-offs are regarded as an extra round. The updates $\Delta R$ are determined for all relevant matches and summed up to get the new Elo ratings. \\
The maximal number of updates in a FIFA World Cup equals the maximal number of matches a team can play, which is seven for the World Cups between 1998 and 2022. \\
In this case, the update frequency will be called \emph{Round}.
\item
\emph{Updating between stages}: The Elo ratings of the continents are fixed at the beginning of each group stage, as well as at the beginning of each knockout round. The inter-continental play-offs count as an extra stage. The updates $\Delta R$ are determined for all relevant matches and summed up to get the new Elo ratings. \\
The maximal number of updates in a FIFA World Cup equals the number of knockout rounds plus the number of group stages, which is five for the World Cups between 1998 and 2022. \\
In this case, the update frequency will be called \emph{Stage}.
\item
\emph{Updating between tournaments}: The Elo ratings of the continents are fixed at the beginning of the inter-continental play-offs. The updates $\Delta R$ are determined for all relevant matches until the end of the World Cup and summed up to get the new Elo ratings. \\
The maximal number of updates in a FIFA World Cup equals one. \\
In this case, the update frequency will be called \emph{4-Year}.
\end{enumerate}

\subsection{The role of match schedule} \label{Sec43}

Take two teams, the underdog $i$ and the favourite $j$ whose initial Elo ratings $R_i^{(0)}$ and $R_j^{(0)}$ differ by $D$: $R_j^{(0)} - R_i^{(0)} = D > 0$. Suppose that they play two matches of equal importance $I$.

If the first game is won by $i$, and the second is won by $j$ such that the ratings are updated between the games, then the change in the Elo ratings after the first game is
\[
\Delta_1^{(ij)} = \Delta R_i = - \Delta R_j = \frac{I}{1+10^{-D/600}}.
\]
Consequently, the difference between the Elo ratings will be $D_1^{(ij)} = D - 2\Delta_1^{(ij)}$, and the change in the Elo ratings after the second game is
\[
\Delta_2^{(ij)} = -\Delta R_i = \Delta R_j = \frac{I}{1+10^{D_1^{(ij)}/600}}.
\]

If the first game is won by $j$, and the second is won by $i$ such that the ratings are updated between the games, then the change in the Elo ratings after the first game is
\[
\Delta_1^{(ji)} = - \Delta R_i = \Delta R_j = \frac{I}{1+10^{D/600}}.
\]
Consequently, the difference between the Elo ratings will be $D_1^{(ji)} = D + 2\Delta_1^{(ji)}$, and the change in the Elo ratings after the second game is
\[
\Delta_2^{(ji)} = \Delta R_i = - \Delta R_j = \frac{I}{1+10^{-D_1^{(ji)}/600}}.
\]

Since $\Delta_1^{(ij)} < \Delta_2^{(ji)}$ due to $D < D_1^{(ji)}$ and $\Delta_2^{(ij)} > \Delta_1^{(ji)}$ due to $D_1^{(ij)} < D$, the implication $\Delta^{(ij)} = \Delta_1^{(ij)} - \Delta_2^{(ij)} < -\Delta_1^{(ji)} + \Delta_2^{(ji)} = \Delta^{(ji)}$ holds, which guarantees that the difference between the Elo ratings is smaller in the second scenario when the second game is won by team $i$ rather than the first.

For instance, with $D=50$ and $I=50$, $\Delta_1^{(ij)} \approx 27.39$, $\Delta_2^{(ij)} \approx 25.23$, $\Delta_1^{(ji)} \approx 22.61$, and $\Delta_2^{(ji)} \approx 29.52$, respectively. Hence, $\Delta^{(ij)} = 2.16$ but $\Delta^{(ji)} = 6.91$, that is, the Elo ratings of teams $i$ and $j$ become closer if team $i$ wins the second match.

This property shows a crucial advantage of our Elo-based approach: it automatically accounts for the order of matches, decreasing the influence of games played in the distant past compared to recent games. Thus, no extra discounting factor is needed.

It is also clear that the frequency of updating is related to the effect of the schedule. In particular, the order of the games does not influence the changes in the Elo ratings if the ratings are not updated between the two games. In the example above, $\Delta = \Delta R_i = - \Delta R_j \approx 4.78$ both in the first (when team $i$ wins the first match) and second (when team $i$ wins the second match) scenarios.

\subsection{Separating a set of top teams} \label{Sec44}

If a continent contains a national team regularly qualifying for FIFA World Cups, it becomes questionable to what extent its performance should be taken into account for the quota deserved by the continent. For instance, Brazil has played in all 18 FIFA World Cups since 1954, hence, Brazil is not a ``marginal'' country whose chance of qualification can be meaningfully improved if the number of CONMEBOL slots is increased. Furthermore, as Brazil has the best records in FIFA World Cups (5 titles and regular participation in at least the quarterfinals), its matches will strongly increase the CONMEBOL share. On the other hand, the Netherlands has played only in nine FIFA World Cups (50\%) since 1954. But it has played three times in the finals and five times in the semifinals, thus, increasing the berths of UEFA can lead to the qualification of a quite competitive national team such as the Netherlands.

Ideally, the allocation rule should be determined by the performance of teams that have an ``average'' probability of qualifying for the FIFA World Cup. The weakest teams do not appear in our database as they are eliminated in the qualification tournaments. The top teams can be separated into a ``seeded'' set, a kind of ``extra'' continent.

Inspired by this idea, two sets of seeded countries have been chosen based on the number of appearances in the FIFA World Cup from 1954 to 2022.
Seeded set S1 contains four countries:
\begin{itemize}
\item 
Argentina (CONMEBOL, 16/18);
\item
Brazil (CONMEBOL, 18/18);
\item
England (UEFA, 15/18);
\item
Germany/West Germany (UEFA, 18/18).
\end{itemize}
Seeded set S2 consists of four additional national teams compared to Seeded set S1:
\begin{itemize}
\item 
France (UEFA, 13/18);
\item
Italy (UEFA, 15/18);
\item
Mexico (CONCACAF, 15/18);
\item
Spain (UEFA, 14/18).
\end{itemize}
No other country has played in more than 12 FIFA World Cups during this period. Although England, Italy, and Mexico have similarly failed to qualify for three tournaments, Italy has missed the last two in 2018 and 2022. Mexico is not considered in the first set because its seeding strongly decreases the number of matches played by CONCACAF teams (as we will see later in Table~\ref{Table3}), and it has never played in the semifinals.
This division may be somewhat arbitrary but it makes sense in our opinion and provides a useful approach for checking the robustness of the allocation.

If some seeded countries are separated, the quota of the confederations is determined without the matches played between their national teams (e.g.\ if eight countries are seeded, the result of Argentina vs.\ Spain does not count at all). However, performances against these top teams are taken into consideration, for instance, an extra win (loss) against Brazil increases (decreases) the number of slots deserved.
Furthermore, the total number of slots to be allocated is reduced by the number of seeds (0/4/8) but these berths are added to the corresponding confederation at the end. For example, if CONMEBOL receives $x$ berths with the Seeded set S1 or S2, then the CONMEBOL quota will be $x+2$ since both sets contain two South American nations.

\subsection{The allocation rule} \label{Sec45}

For any two teams, equation~\eqref{eq_win_probability} immediately gives how many times team $i$ is better than team $j$ by $a_{ij} = W_{ij}^E / W_{ji}^E$.
These values are transitive according to the following result.

\begin{proposition} \label{Prop1}
Let $a_{ij} = W_{ij}^E / W_{ji}^E$ mean how many times team $i$ is more likely to win against team $j$ than vice versa. These paired comparisons are transitive, tha is, $a_{ik} = a_{ij} a_{jk}$ for any $i,j,k$.
\end{proposition}

\begin{proof}
Due to formula \eqref{eq_win_probability}:
\begin{eqnarray} \label{eq_transitivity}
a_{ij} \times a_{jk} & = & \frac{W_{ij}^E}{W_{ji}^E} \times \frac{W_{jk}^E}{W_{kj}^E} = \frac{1+10^{\left( R_i - R_j \right) / 600}}{1+10^{- \left( R_i - R_j \right) / 600}} \times \frac{1+10^{\left( R_j - R_k \right) / 600}}{1+10^{- \left( R_j - R_k \right) / 600}} = \nonumber \\
& = & \frac{1 + 10^{\left( R_i - R_j \right) / 600} + 10^{\left( R_j - R_k \right) / 600} + 10^{\left( R_i - R_k \right) / 600}}{10^{-\left( R_i - R_k \right) / 600} + 10^{-\left( R_j - R_k \right) / 600} + 10^{-\left( R_i - R_j \right) / 600} + 1} = \nonumber \\
& = & 10^{\left( R_i - R_k \right) / 600} = \frac{1 + 10^{\left( R_i - R_k \right) / 600}}{1 + 10^{-\left( R_i - R_k \right) / 600}} = \frac{W_{ik}^E}{W_{ki}^E} = a_{ik}.
\end{eqnarray}
\end{proof}

Therefore, the total number of available slots ($48-4/3$ since the OFC quota is fixed at $4/3$) can be easily allocated according to the pairwise comparisons $a_{ij}$.
If a continent $k$ is chosen arbitrarily, the number of slots deserved by continent $i$ is computed as
\begin{equation} \label{eq_quota}
q_i = \frac{a_{ik}}{\sum_{j=1}^n a_{jk}} \times \left( 48 - \frac{4}{3} - \lvert S \lvert \right) + \lvert S_i \lvert,
\end{equation}
where $\lvert S \lvert$ is the number of seeded nations and $\lvert S_i \lvert$ is the number of seeded nations from continent $i$.
Transitivity ensures that $q_i$ is independent of $k$.

Let us see an illustration of how the slots for the continents are determined.

\begin{example}
Take all matches played until the 2002 FIFA World Cup, including the championship this particular year, too. Consider Updating between rounds and Seeded set S2. The Elo ratings of the confederations AFC, CAF, CONCACAF, CONMEBOL, and UEFA are 1576.56, 1734.71, 1574.12, 1590.36, and 1806.89, respectively. Let $k$ be the confederation AFC. Thus, the number of slots deserved by the CAF is
\[
\frac{1 + 10^{(1576.56-1734.71)/600}}{1 + 10^{-(1576.56-1734.71)/600}} =  10^{(1734.71-1576.56)/600} \approx 1.83
\]
times the number of slots deserved by the AFC as can be seen from formula \eqref{eq_transitivity}. The corresponding ratios for the CONCACAF, CONMEBOL, and UEFA are 0.99, 1.05, and 2.42, respectively.

The number of slots for the AFC is
\[
\frac{48 - 4/3 - 8}{1 + 1.83 + 0.99 + 1.05 + 2.42} \approx 5.3.
\]
Analogously, the number of slots for the UEFA is
\[
2.42 \times \frac{48 - 4/3 - 8}{1 + 1.83 + 0.99 + 1.05 + 2.42} + 5 \approx 17.82
\]
as there are five seeded UEFA countries according to Section~\ref{Sec44}.
\end{example}

\subsection{The maximal number of slots for a confederation}  \label{Sec46}

The allocation rule in Section~\ref{Sec45} does not contain any restriction on the inequality of the allocation, that is, a particular confederation might get any number of slots between zero and the total available slots. This can be a problem in the case of CONMEBOL (the strongest continent according to Table~\ref{Table3}), which has only 10 members. Consequently, in order to guarantee the competitiveness of the CONMEBOL qualification, we have decided to maximise its quota at 8, which allows direct qualification for the top seven South American countries and qualification for the play-offs for the next two.  


\section{Results and discussion} \label{Sec5}

\begin{table}[t!]
  \centering
  \caption{The results of matches by pairs of confederations}
  \label{Table3}

\begin{subtable}{\textwidth}
  \caption{Seeded set S0}
  \label{Table3a}
    \rowcolors{1}{}{gray!20}
\centerline{
    \begin{tabularx}{1.05\textwidth}{lCCCCcc} \toprule
          & AFC   & CAF   & CONC & CONM & UEFA  & $\sum$ \\ \bottomrule
    AFC   & 0 (0) & 6 (5) & 2 (3) & 3 (4) & 9 (12) & 20 (24) \\
    CAF   & 5 (5) & 0 (0) & 2 (2) & 2 (2) & 16 (19) & 25 (28) \\
    CONCACAF & 6 (3) & 2 (2) & 2 (0) & 4 (4) & 10 (16) & 24 (25) \\
    CONMEBOL & 17 (4) & 15 (2) & 14 (4) & 15 (1) & 77 (31) & 138 (42) \\
    UEFA  & 41 (12) & 41 (19) & 38 (16) & 68 (31) & 173 (30) & 361 (108) \\ \hline
    $\sum$ & 69 (24) & 64 (28) & 58 (25) & 92 (42) & 285 (108) & \textbf{568 (129)} \\ \bottomrule
    \end{tabularx}
}
\end{subtable}

\vspace{0.25cm}
\begin{subtable}{\textwidth}
  \caption{Seeded set S1}
  \label{Table3b}
    \rowcolors{1}{}{gray!20}
\centerline{
    \begin{tabularx}{1.05\textwidth}{lCCCCccc} \toprule
          & AFC   & CAF   & CONC & CONM & UEFA  & Seeded S1 & $\sum$ \\ \bottomrule
    AFC   & 0 (0) & 6 (5) & 2 (3) & 2 (4) & 8 (12) & 2 (0) & 20 (24) \\
    CAF   & 5 (5) & 0 (0) & 2 (2) & 1 (2) & 15 (16) & 2 (3) & 25 (28) \\
    CONCACAF & 6 (3) & 2 (2) & 2 (0) & 4 (3) & 9 (14) & 1 (3) & 24 (25) \\
    CONMEBOL & 9 (4) & 6 (2) & 4 (3) & 2 (0) & 18 (14) & 2 (3) & 41 (26) \\
    UEFA  & 36 (12) & 34 (16) & 29 (14) & 28 (14) & 110 (14) & 42 (26) & 279 (96) \\
    Seeded S1 & 13 (0) & 16 (3) & 19 (3) & 21 (3) & 86 (26) & 24 (5) & 179 (40) \\ \hline
    $\sum$ & 69 (24) & 64 (28) & 58 (25) & 58 (26) & 246 (96) & 73 (40) & \textbf{568 (129)} \\ \toprule
    \end{tabularx}
}
\end{subtable}

\vspace{0.25cm}
\begin{subtable}{\textwidth}
  \caption{Seeded set S2}
  \label{Table3c}
\centerline{
\begin{threeparttable}
    \rowcolors{1}{}{gray!20}
    \begin{tabularx}{1.05\textwidth}{lCCCCccc} \toprule
          & AFC   & CAF   & CONC & CONM & UEFA  & Seeded S2 & $\sum$ \\ \bottomrule
    AFC   & 0 (0) & 6 (5) & 2 (3) & 2 (4) & 6 (10) & 4 (2) & 20 (25) \\
    CAF   & 5 (5) & 0 (0) & 1 (0) & 1 (2) & 12 (16) & 6 (5) & 25 (28) \\
    CONCACAF & 3 (3) & 1 (0) & 0 (0) & 2 (1) & 3 (7) & 2 (4) & 11 (15) \\
    CONMEBOL & 9 (4) & 6 (2) & 4 (1) & 2 (0) & 15 (6) & 5 (13) & 41 (26) \\
    UEFA & 28 (10) & 25 (16) & 16 (7) & 19 (6) & 56 (9) & 49 (27) & 193 (75) \\
    Seeded S2  & 24 (2) & 26 (5) & 16 (4) & 32 (13) & 107 (27) & 73 (15) & 278 (66) \\ \hline
    $\sum$ & 69 (25) & 64 (28) & 39 (15) & 58 (26) & 199 (75) & 139 (66) & \textbf{568 (129)} \\ \toprule
    \end{tabularx}
\begin{tablenotes} \footnotesize
\item
Seeded set S0 means that each national team is assigned to its continent.
\item
Seeded set S1 contains Argentina (CONMEBOL), Brazil (CONMEBOL), England (UEFA), Germany (UEFA).
\item
Seeded set S2 contains Argentina (CONMEBOL), Brazil (CONMEBOL), England (UEFA), France (UEFA), Germany (UEFA), Italy (UEFA), Mexico (CONCACAF), Spain (UEFA).
\item
CONC stands for CONCACAF; CONM stands for CONMEBOL.
\item
The cells show the number of matches won by the confederation in the row against the confederation in the column.
\item
Matches decided in a penalty shootout are considered standard wins and losses.
\item
Two-legged play-offs are counted as two matches, analogous to the FIFA World Ranking.
\item
The number of draws is given in parenthesis.
\item
Matches played in the last round of the (first) group stage are not included.
\end{tablenotes}
\end{threeparttable}
}
\end{subtable}

\end{table}

The outcomes of matches in our database are summarised in Table~\ref{Table3}, separately for the three sets of seeded teams. As expected, two confederations have won the majority of their matches without counting draws, CONMEBOL (60\%) and UEFA (55.88\%) (Table~\ref{Table3a}). However, this does not hold if the best countries are separated from these confederations: the balance of CONMEBOL becomes negative under both Seeded sets S1 and S2 (Table~\ref{Table3b}), and the balance of UEFA becomes negative under Seeded set S2 (Table~\ref{Table3c}). Unsurprisingly, the national teams in both Seeded sets S1 (71.03\%) and S2 (66.67\%) have won most of their matches without counting draws.

\input{Figure_historical_slot_allocation}

Figures~\ref{Fig1} and \ref{Fig2} present how the number of slots would have evolved for the five confederations if the database had been finished after one of the last eight FIFA World Cups. Figure~\ref{Fig1} compares Seeded sets S0 and S1 under the three different update frequencies, while Figure~\ref{Fig2} repeats this analysis for four (S1) and eight (S2) seeded nations. CONMEBOL strongly benefits if the performance of Argentina and Brazil are taken into account as can be seen in Figure~\ref{Fig1}. Nonetheless, South America almost always receives a quota above 10 (the number of its members) if the set of matches is extended at least to the 2010 FIFA World Cup. In addition, the minimum is still above 6.5, which exceeds the number of berths provided by FIFA (see Table~\ref{Table2}).

Contrarily, UEFA usually receives a higher number of slots if some countries are seeded according to Figure~\ref{Fig1}. Figure~\ref{Fig2} uncovers that separating the top five European nations (Seeded set S2) makes the share of UEFA less volatile. Among the other three continents, the performance of African teams seems to be the least stable: the number of CAF slots varies approximately from 8 to 13.5 if the analysis is finished after the 2002 FIFA World Cup, while it remains below 3.5 if the last tournament is the 2018 edition. The results for the AFC (between 2 and 6) and CONCACAF (from 3 to 9.5) are somewhat more robust. Our model implies that CONCACAF should get more slots than AFC in almost every scenario, which is in stark contrast with the official FIFA policy.

However, Figures~\ref{Fig1} and \ref{Fig2} reveal that the length of the sample has a powerful effect on the allocation of slots. Naturally, this is the most visible if the ratings are updated only after each World Cup (the bottom panels) when the teams of a given continent can play 10 or even 20 matches with the same Elo, while the FIFA World Ranking is updated after every single game. A potential remedy can be a lower value for match importance $I$ (see Section~\ref{Sec41}), but we have not wanted to arbitrarily modify the official formula. Another opportunity is to distribute only a fraction of the total slots on the basis of historical performances. For example, one can start from the status quo allocation of 31 slots (one reserved for the host) used between 2006 and 2022 (Table~\ref{Table2}), and allocate the residual 16 slots by our method.

Last but not least, it is interesting to see in Figures~\ref{Fig1} and \ref{Fig2} that the performance of CONMEBOL and UEFA has certainly not declined after 1998. Furthermore, the CONMEBOL quota somewhat increases under Seeded set S2 if the sample is finished later. Hence, the FIFA policy of supporting emerging nations has had questionable results, and the overall quality of the 2026 FIFA World Cup will probably be lower due to the stronger presence of AFC and CAF countries.

\begin{table}[t!]
  \centering
  \caption{Slot allocations in the baseline model (Update frequency: Round; Seeded set S2)}
  \label{Table4}
\centerline{
\begin{threeparttable}
    \rowcolors{1}{gray!20}{}
    \begin{tabularx}{1.08\textwidth}{l CCCCC CCCc} \toprule \hiderowcolors
    \multirow{2}[0]{*}{Confederation} & \multicolumn{8}{c}{Results are taken into account until}      & Official \\
          & 1994  & 1998  & 2002  & 2006  & 2010  & 2014  & 2018  & 2022  & allocation \\ \bottomrule \showrowcolors
    AFC   & 3.69  & \underline{2.81}  & 5.30  & \textbf{5.83}  & 5.19  & 3.88  & 3.48  & 4.48  & 8.33 \\
    CAF   & 6.59  & 6.18  & \textbf{9.72}  & 7.16  & 6.19  & 5.54  & \underline{3.96}  & 7.43  & 9.33 \\
    CONCACAF & 5.38  & \underline{3.97}  & 6.25  & 6.26  & 5.32  & \textbf{7.24}  & 4.81  & 5.33  & 3.67+3 \\
    CONMEBOL & \underline{7.10}  & 7.37  & 7.59  & \textbf{8}     & \textbf{8}     & \textbf{8}     & \textbf{8}     & \textbf{8}     & 6.33 \\
    UEFA  & 23.91 & \textbf{26.34} & \underline{17.82} & 19.42 & 21.98 & 22.01 & 26.43 & 21.43 & 16 \\ \toprule
    \end{tabularx}
\begin{tablenotes} \footnotesize
\item
Seeded set S2 contains Argentina (CONMEBOL), Brazil (CONMEBOL), England (UEFA), France (UEFA), Germany (UEFA), Italy (UEFA), Mexico (CONCACAF), Spain (UEFA).
\item
The CONMEBOL quota is maximised at 8 since there are 10 CONMEBOL members.
\item
The last round of group matches is disregarded.
\item
For each continental zone, the maximal number of slots is written in \textbf{bold}, while the minimal number of slots is \underline{underlined} if the end of the sample is between 1994 and 2022.
\end{tablenotes}
\end{threeparttable}
}
\end{table}

According to Figures~\ref{Fig1} and \ref{Fig2}, as well as theoretical considerations, the results under Seeded set S2 should be the baseline. Furthermore, in order to minimise the number of matches played without rating changes, the update frequency Round should be chosen. The findings from this baseline model are summarised in Table~\ref{Table4}, taking into account that South America cannot get more than eight slots (these extra slots have not been reallocated in the top right panel of Figure~\ref{Fig2}).

Hence, CONMEBOL deserves at least 7 and UEFA deserves at least 18 slots, both numbers being higher than the berths provided for them in the 2026 FIFA World Cup. Again, the CONCACAF zone seems to be more competitive than the AFC zone in historical FIFA World Cups, which is not reflected in the official allocation. This is a remarkable fact since the separation of Mexico in the Seeded set S2 almost always decreases the CONCACAF quota as can be seen in Figure~\ref{Fig2}. The performance of the CAF teams is the most volatile but the CAF is robustly entitled to a higher share than the AFC in line with the FIFA policy.

\begin{table}[t!]
  \centering
  \caption{Slot allocations for the 2026 FIFA World Cup based on different methods}
  \label{Table5}
\centerline{
\begin{threeparttable}
    \rowcolors{1}{gray!20}{}
    \begin{tabularx}{1.12\textwidth}{l CCC CCC CCC} \toprule \hiderowcolors
    Seeding & \multicolumn{3}{c}{No seeded teams} & \multicolumn{3}{c}{Seeded set S1} & \multicolumn{3}{c}{Seeded set S2} \\ 
    Update frequency & Round & Stage & 4-year & Round & Stage & 4-year & Round & Stage & 4-year \\ \bottomrule \showrowcolors
    AFC   & 4.77  & 5.30  & 5.89  & 3.82  & 4.10  & 4.54  & 4.48  & 4.89  & 5.34 \\
    CAF   & 7.60  & 7.52  & 8.39  & 6.40  & 6.16  & 6.65  & 7.43  & 7.13  & 7.60 \\
    CONCACAF & 6.21  & 6.65  & 8.07  & 5.26  & 5.65  & 6.68  & 5.33  & 5.36  & 5.50 \\
    CONMEBOL & 8     & 8     & 8     & 8     & 8     & 8     & 8     & 8     & 8 \\
    OFC   & 1.33  & 1.33  & 1.33  & 1.33  & 1.33  & 1.33  & 1.33  & 1.33  & 1.33 \\
    UEFA  & 20.09 & 19.19 & 16.31 & 23.19 & 22.75 & 20.80 & 21.43 & 21.29 & 20.23 \\ \bottomrule
    \end{tabularx}
\begin{tablenotes} \footnotesize
\item
Seeded set S1 contains Argentina (CONMEBOL), Brazil (CONMEBOL), England (UEFA), Germany (UEFA).
\item
Seeded set S2 contains Argentina (CONMEBOL), Brazil (CONMEBOL), England (UEFA), France (UEFA), Germany (UEFA), Italy (UEFA), Mexico (CONCACAF), Spain (UEFA).
\item
The CONMEBOL quota is maximised at 8 since there are 10 CONMEBOL members.
\item
The OFC quota is fixed at 1.33 and does not come from our method.
\item
The last round of group matches is disregarded.
\end{tablenotes}
\end{threeparttable}
}
\end{table}

Table~\ref{Table5} focuses on the slot allocation for the 2026 FIFA World Cup based on all available information. The rule of the maximal number of slots (Section~\ref{Sec46}) applies to the CONMEBOL, resulting in a quota of  eight. Both separating the nations with the highest number of appearances and choosing the frequency of Elo updates have non-negligible effects, as has already been uncovered in Figures~\ref{Fig1} and \ref{Fig2}. Rarer updates are associated with fewer slots for UEFA and more for all other continents (except for OFC and CONMEBOL).

On the other hand, the results are not that straightforward if the number of seeded nations increases. First (changing from Seeded set S0 to S1), the slots of UEFA grow substantially. This is probably due to the high number of competitive nations in Europe even without England and Germany, which is amplified by the seeding of Argentina and Brazil, the best South American teams. Then---between Seeded sets S1 and S2---the share of UEFA slightly decreases as even more successful nations are placed in the seeded set.

\begin{table}[t!]
  \centering
  \caption{The effect of including the last round of group matches on the allocation of slots}
  \label{Table6}
\centerline{
\begin{threeparttable}
    \rowcolors{1}{gray!20}{}
    \begin{tabularx}{1.12\textwidth}{l RRR RRR RRR} \toprule \hiderowcolors
    Seeding & \multicolumn{3}{c}{No seeded teams} & \multicolumn{3}{c}{Seeded set S1} & \multicolumn{3}{c}{Seeded set S2} \\ 
    Update freq.\ & Round & Stage & 4-year & Round & Stage & 4-year & Round & Stage & 4-year \\ \bottomrule \showrowcolors
    AFC   & 1.41  & 1.98  & 2.78  & 1.76  & 2.36  & 3.72  & 1.32  & 1.74  & 2.46 \\
    CAF   & 3.39  & 3.96  & 6.21  & 3.55  & 4.29  & 7.05  & 2.84  & 3.50  & 4.70 \\
    CONCACAF & $-$0.27 & $-$0.36 & $-$1.20 & 0.17  & 0.01  & $-$0.05 & $-$0.32 & $-$0.31 & $-$0.48 \\
    UEFA  & $-$4.52 & $-$5.58 & $-$7.80 & $-$5.48 & $-$6.66 & $-$10.71 & $-$3.85 & $-$4.94 & $-$6.69 \\ \toprule
    \end{tabularx}
\begin{tablenotes} \footnotesize
\item
Seeded set S1 contains Argentina (CONMEBOL), Brazil (CONMEBOL), England (UEFA), Germany (UEFA).
\item
Seeded set S2 contains Argentina (CONMEBOL), Brazil (CONMEBOL), England (UEFA), France (UEFA), Germany (UEFA), Italy (UEFA), Mexico (CONCACAF), Spain (UEFA).
\end{tablenotes}
\end{threeparttable}
}
\end{table}

Finally, Table~\ref{Table6} provides a kind of sensitivity analysis by including the last round of group matches, which has notable effects on the results. The CONMEBOL and OFC are not investigated as the former confederation always receives the maximal number of slots allowed and the quota of the latter is fixed. However, depending on the model variant considered, UEFA would lose at least 3.8 and up to 10.7 slots due to considering these matches. At the same time, AFC would gain between 1.3 and 3.7 slots, while CAF between 2.8 and 7.
Clearly, the AFC and CAF teams won several matches in the last round of the group stage against European countries that have probably already qualified for the knockout stage. No similar trend can be observed for South American teams, which are perhaps less prone to strategic behaviour. Table~\ref{Table6} suggests that the problem of incentives cannot be neglected in the analysis of FIFA World Cup group matches.

The results above are worth comparing with the findings of \citet{KrumerMoreno-Ternero2023}, although this is not so straightforward since they also investigate five methods. Furthermore, an important novelty of the current study is our ``look into the past'' in the sense that slot allocations are investigated if the expansion would have been decided before.
We see two main differences:
\begin{itemize}
\item 
\citet{KrumerMoreno-Ternero2023} do not imply such a high quota for South America;
\item
AFC is entitled to more slots than CONCACAF according to \citet{KrumerMoreno-Ternero2023}, in line with the policy of FIFA, while a reversed order emerges from our methods.
\end{itemize}
However, our proposal for a higher number of UEFA slots is fully supported by \citet{KrumerMoreno-Ternero2023}.

The discrepancies are obviously caused by the different approaches of the two studies. \citet{KrumerMoreno-Ternero2023} mainly start from the FIFA World Ranking and the World Football Elo Ratings: they count the number of countries among the top 48, and the average annual number of teams ranked 1--31 and 32--48. This approach implicitly assumes that both measures provide a reliable ranking of at least the top 48 teams, which is not necessarily the case.

In particular, the ratings of South American teams are likely too low. The CONMEBOL region is different from the other ones because it has the lowest number of teams (only ten) and its qualification tournament is probably the least predictable. The South American qualifiers for the FIFA World Cup are organised as a double round-robin tournament, where each team plays 18 matches \citep{DuranGuajardoSaure2017}, Hence, the majority of games played by CONMEBOL teams are against teams from the same continent \citep{Price2021}. The results of these matches do not affect the average Elo of the confederation. In addition, CONMEBOL contains only ten nations and even the weakest perform relatively well: the last team scored 12 (10) points out of 54 in the qualification for the 2018 (2022) FIFA World Cup. Therefore, in contrast to the other confederations, the best countries in the CONMEBOL zone are not able to collect so many points against bottom teams having no reasonable chance to qualify.

Naturally, our allocation rules should not necessarily be preferred to the ones considered by \citet{KrumerMoreno-Ternero2023}. The methodology proposed here is not able to take the expected performance of teams that have never (or rarely) qualified for the FIFA World Cup into account. If a confederation is dominated by some teams performing well in previous inter-continental matches, then its quota can be relatively high even if the next nation is potentially too weak for the FIFA World Cup.

However, if the somewhat orthogonal approaches of us and \citet{KrumerMoreno-Ternero2023} point in the same direction (UEFA should receive additional slots), the decision-makers could have few arguments against accepting the proposal.

\section{Concluding remarks} \label{Sec6}

Inspired by the recent expansion of the FIFA World Cup to 48 teams, the current work has proposed allocation rules to distribute tournament slots among different sets of teams. Our approach adapts the methodology of the official FIFA World Ranking, a rating of national teams, to compare the performance of continents in the FIFA World Cups and their inter-continental play-offs since 1954. Various alternatives exist based on the frequency of updating, as well as the length of the sample, and the set of teams seeded. 
The results imply that more European and South American teams should play in future FIFA World Cups. Furthermore, CONCACAF deserves more slots than AFC.

Unfortunately, the current qualification system of the FIFA World Cup is based neither on ensuring the participation of the best 32 teams in the world, nor does it fairly allocate slots according to a reasonable metric \citep{StoneRod2016}. Some national teams can gain from moving across confederations \citep{Csato2023c}, or even from deliberately losing in their own tournament \citep{Csato2022a}.
The proposed ranking of FIFA confederations is able to mitigate these problems to some extent. Furthermore, such an allocation of continental slots can create stronger incentives to perform: for example, taking the games played in the last round of the group stage into account would mean additional slots for the continental zone in the case of winning.

To summarise, choosing a fairer, more transparent slot allocation method would be a socially responsible obligation of FIFA. Hopefully, at least some berths will be determined by well-defined rules in the future. The recent improvement of the FIFA World Ranking \citep{FIFA2018c}, which has eliminated the main weaknesses of the previous formula \citep{CeaDuranGuajardoSureSiebertZamorano2020, Csato2021a, Kaminski2022, LasekSzlavikGagolewskiBhulai2016}, shows that FIFA is open to suggestions from the academic community.

The proposed approach of rating sets of teams based on historical matches between them can be used in other contexts than the FIFA World Cup. As mentioned by \citet{KrumerMoreno-Ternero2023}, the most straightforward case seems to be the European club competitions organised by the Union of European Football Associations (UEFA). Here, the number of slots provided for each association depends on the UEFA country ranking, the average number of points collected by the participating clubs over the last five seasons \citep{Csato2022b}. However, the interaction of these international contests and the national leagues \citep{GunerHamidiSahneh2023, RappaiFuresz2024} can justify different allocation rules. For instance, there is robust evidence that the substantial income from the UEFA club competitions reduces the intensity of competition in domestic leagues and leads to the dominance of some teams over the long term \citep{PawlowskiBreuerHovemann2010, Peeters2011}. Since improving competitive balance is an important aim of the UEFA \citep{Gyimesi2024}, it would be unfavourable to apply a strictly performance-based rule without any compensation mechanism. Note that this consideration is much less relevant for national teams as they cannot increase their strength by buying high-ability players on the market.

\section*{Acknowledgements}

The paper develops a study of \emph{L\'aszl\'o Marcell Kiss} that has been presented at the Scientific Students' Associations of Corvinus University of Budapest in 2023. \\
\emph{Andr\'as Gyimesi} and four anonymous reviewers have given useful remarks. \\
The research was supported by the National Research, Development and Innovation Office under Grant FK 145838.

\bibliographystyle{apalike} 
\bibliography{All_references}

\end{document}